\documentclass[10pt]{article}
\usepackage[utf8]{inputenc}
\usepackage{amsmath, amssymb, amscd, amsthm, amsfonts, bbm, geometry}
\usepackage{fullpage}

\usepackage{graphicx}
\usepackage{hyperref}
\hypersetup{
    colorlinks=true,
    allcolors=purple,
}

\usepackage{braket}
\usepackage[dvipsnames]{xcolor}
\usepackage{algorithmicx}
\usepackage{algorithm}
\usepackage{algpseudocode}
\usepackage{qcircuit}
\usepackage{stmaryrd}
\usepackage{mleftright}
\mleftright
\usepackage{caption}
\usepackage{subcaption}
\usepackage{titling}

\theoremstyle{definition}
\newtheorem{theorem}[equation]{Theorem}
\newtheorem{proposition}[equation]{Proposition}
\newtheorem{corollary}[equation]{Corollary}
\newtheorem{lemma}[equation]{Lemma}

\newtheorem{remark}[equation]{Remark}

\DeclareMathOperator{\id}{id}

\DeclareMathOperator{\diag}{diag}
\DeclareMathOperator{\poly}{poly}

\DeclareMathOperator{\SWAP}{SWAP}

\DeclareMathOperator{\tG}{G}
\DeclareMathOperator{\tH}{H}
\DeclareMathOperator{\tZ}{Z}

\newcommand{\abs}[1]{\left\lvert#1\right\rvert}
\newcommand{\eps}{\varepsilon}

\newcommand{\gb}{\beta}
\newcommand{\g}{\gamma}
\newcommand{\gd}{\delta}
\newcommand{\gw}{\omega}

\newcommand{\gl}{\lambda}
\newcommand{\gL}{\Lambda}
\newcommand{\gs}{\sigma}
\newcommand{\gT}{\Theta}

\newcommand{\bP}{\mathbb{P}}

\newcommand{\cE}{\mathcal{E}}

\newcommand{\cO}{\mathcal{O}}
\DeclareFontFamily{U}{jkpmia}{}
\DeclareFontShape{U}{jkpmia}{m}{it}{<->s*jkpmia}{}
\DeclareFontShape{U}{jkpmia}{bx}{it}{<->s*jkpbmia}{}
\DeclareMathAlphabet{\mathfrak}{U}{jkpmia}{m}{it}
\SetMathAlphabet{\mathfrak}{bold}{U}{jkpmia}{bx}{it}
\newcommand{\fA}{\mathfrak{A}}
\newcommand{\fB}{\mathfrak{B}}
\newcommand{\fC}{\mathfrak{C}}
\newcommand{\fD}{\mathfrak{D}}
\newcommand{\fE}{\mathfrak{E}}
\newcommand{\fF}{\mathfrak{F}}
\newcommand{\fG}{\mathfrak{G}}
\newcommand{\norm}[1]{\left\|#1\right\|}
\newcommand{\lpr}[1]{\left(#1\right)}
\newcommand{\lbr}[1]{\left[#1\right]}
\newcommand{\tr}{\operatorname{tr}}

\newcommand{\appref}[1]{\hyperref[#1]{Appendix \ref{#1}}}
\newcommand{\claimref}[1]{\hyperref[#1]{Claim \ref{#1}}}
\newcommand{\conjref}[1]{\hyperref[#1]{Conjecture \ref{#1}}}
\newcommand{\corref}[1]{\hyperref[#1]{Corollary \ref{#1}}}
\newcommand{\defref}[1]{\hyperref[#1]{Definition \ref{#1}}}
\newcommand{\rmkref}[1]{\hyperref[#1]{Remark \ref{#1}}}
\newcommand{\exref}[1]{\hyperref[#1]{Example \ref{#1}}}
\newcommand{\figref}[1]{\hyperref[#1]{Figure \ref{#1}}}
\newcommand{\hwref}[1]{\hyperref[#1]{Homework \ref{#1}}}
\newcommand{\lemref}[1]{\hyperref[#1]{Lemma \ref{#1}}}
\newcommand{\probref}[1]{\hyperref[#1]{Problem \ref{#1}}}
\newcommand{\propref}[1]{\hyperref[#1]{Proposition \ref{#1}}}
\newcommand{\secref}[1]{\hyperref[#1]{\S\ref{#1}}}
\newcommand{\tabref}[1]{\hyperref[#1]{Table \ref{#1}}}
\newcommand{\thmref}[1]{\hyperref[#1]{Theorem \ref{#1}}}
\newcommand{\half}{\frac{1}{2}}
\newcommand{\db}[1]{\left\llbracket#1\right\rrbracket}
\newcommand{\bgl}{\boldsymbol\lambda}
\newcommand{\bu}{{\boldsymbol u}}

\newcommand{\QMA}{$\mathsf{QMA}$}
\newcommand{\NP}{$\mathsf{NP}$}
\newcommand{\BQP}{$\mathsf{BQP}$}
\newcommand{\DQC}{$\mathsf{DQC}$}
\renewcommand{\P}{$\mathsf{P}$}

\expandafter\let\expandafter\originald\csname\encodingdefault\string\d\endcsname
\DeclareRobustCommand*\d
    {\ifmmode\mathop{}\!\mathrm{d}\else\expandafter\originald\fi}

\newcommand{\arxiv}[1]{\href{https://arxiv.org/abs/#1}{arXiv:#1}}

\usepackage{tikz}
\usetikzlibrary{circuits.ee.IEC}
\newcommand\Ground{
\mathbin{\text{\begin{tikzpicture}[circuit ee IEC,yscale=0.6,xscale=0.5]
\draw node[ground,rotate=0,xshift=.65ex] {};
\end{tikzpicture}}}
}

\usepackage{authblk}
\title{A no-go result for pure state synthesis in the \DQC1 model}
\author{Zachary Stier}
\affil{UC Berkeley Department of Mathematics\\\texttt{zstier@berkeley.edu}}
\date{April 2024}

\begin{document}

\maketitle

\begin{abstract}
    We study the problem of state synthesis in the \DQC1 (One Clean Qubit) model of quantum computation, which provides a single pure qubit and $n$ maximally mixed qubits, and after applying any quantum circuit some subset of the qubits are measured or discarded. In the case of discarding, we show that it is impossible to prepare additional pure qubits, and that it is impossible to prepare very low-temperature Gibbs states on additional qubits. In the case of measurements, we show that the probability of synthesizing $m$ additional qubits is bounded by $2^{1-m}$, and that the probability of preparing low-temperature Gibbs states is bounded by $2^{2-m}$. As a consequence, we give a lower-bound the runtime of a recently studied class of repeated interaction quantum algorithms. The techniques used study states and circuits at the level of entries of their respective density and unitary matrices. 
\end{abstract}

\section{Introduction}

In this paper, we are concerned with the problem of {\em state synthesis}. One is provided a quantum computer with a given initial state and a designated output register which is to be left in some prescribed target state, with the remaining qubits (those not in the output register) to be either measured or discarded, as specified. The task is to design a circuit which, from the given input, leaves the output register in the target state; the procedure of inputting pure ancillae, applying a circuit, and measuring or discarding may be iterated a number of times, replacing the measured or discarded qubits by a prescribed fresh state, and the goal is to design a circuit for each iteration. See Figures \ref{fig:1} and \ref{fig:2}
\begin{figure}
    \centering
    \begin{subfigure}[b]{0.4\textwidth}
        \centering
        \leavevmode
        \Qcircuit @C=1.5em @R=1.5em {
        \lstick{\fA} & {/^a} \qw & \multigate{1}{U} \qw & {/^c} \qw & \qw & \meter & {\fC} \\
        \lstick{\fB} & {/^b} \qw & \ghost{U} & {/^d} \qw & \qw & \qw & {\fD} 
        }
        \caption{A measuring \DQC$a$ machine.}
        \label{fig:1}
    \end{subfigure}
    \begin{subfigure}[b]{0.4\textwidth}
        \centering
        \leavevmode
        \Qcircuit @C=1.5em @R=1.5em {
        \lstick{\fA} & {/^a} \qw & \multigate{1}{U} \qw & {/^c} \qw & \qw & \measuretab{\Ground{}} & {\fC} \\
        \lstick{\fB} & {/^b} \qw & \ghost{U} & {/^d} \qw & \qw & \qw & {\fD} 
        }
        \caption{A discarding \DQC$a$ machine.}
        \label{fig:2}
    \end{subfigure}
    \caption{\DQC\ machines studied here.}
\end{figure}
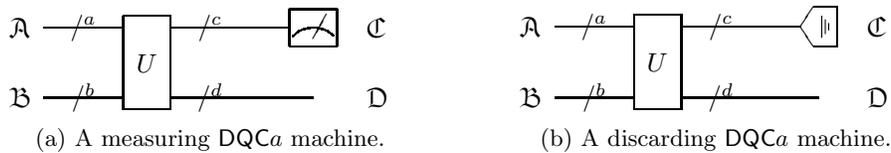
for the singly-iterated setting, where $\fA\sqcup\fB$ is the input register, $\fC$ is the register to be measured or discarded, $\fD$ is the output register, and the task is to find an appropriate $U$. Having efficient means by which to enact state preparation is recognized as central to quantum chemistry \cite{OBK+,LLZ+}, and various examples of algorithms include \cite{ADLH,GTC,LT}. 

One instance of this broad class of problems is state synthesis in \DQC$k$, where the computational model, \DQC$k$, has input with $a=k$ qubits given in the state $\ket{0}\bra{0}_\fA$ and $b=n$ qubits given in the maximally mixed state (i.e.\ the density matrix $\frac{1}{2^n}\id_\fB$), with output on $d=n$ qubits (and $c=k$ qubits measured or discarded). The model \DQC1, also called {\em One Clean Qubit}, is already quite useful; the following application, though not an instance of state preparation, is still of great importance. Given many copies of the controlled-$U$ gate as well as the Hadamard gate $H$, and allowed many runs on the quantum computer, the {\em Hadamard test} (\figref{fig:hadamard}) 
\begin{figure}
    \centering
    \leavevmode
    \Qcircuit @C=1.5em @R=1.5em {
    \ket{0} && \gate{H} & \ctrl{1} & \gate{H} & \meter \\
    \frac{1}{2^n}\id && {/^n} \qw & \gate{U} & \qw & \qw
    }
    \captionsetup{width=.75\textwidth}
    \caption{The Hadamard test, for estimating the trace of the unitary $U$.}
    \label{fig:hadamard}
\end{figure}
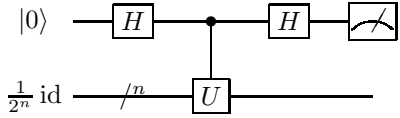
allows one to compute $U$'s trace to precision $\eps$ using $\poly\frac{1}{\eps}$ runs (cf.\ \cite{Lin,KL}). \DQC1 has also found application to computing the Jones polynomial of a knot \cite{SJ}. A nice complexity-theoretic result is that, letting the complexity class \DQC$k$ be the set of decision problems efficiently decidable with a \DQC$k$ machine, \cite{Shepherd} shows that $\mathsf{DQC}1=\mathsf{DQC}\log n$. 

Let a {\em measuring \DQC$k$ machine} be one where $\fC$ is to be measured; there, we will be concerned with probabilities of certain state syntheses. In contrast, letting a {\em discarding \DQC$k$ machine} be one where $\fC$ is to be discarded, we will there give deterministic results. 

Our main result is the following: 
\begin{theorem}[{\thmref{thm:tr}, stated informally}]\label{thm:tr''}
    If $n>k$ then it is impossible to synthesize a $n$-qubit pure state on a discarding \DQC$k$ machine. 
\end{theorem}

(Of course, if $n\le k$ then such a synthesis is immediate by simply using $\SWAP$ gates to ensure that the surplus of given pure qubits are moved to the output register. We actually prove a slightly stronger result, where we are not necessarily constrained to have $a=c$ and $b=d$.) 

This fact has many interesting consequences, including a probabilistic result for measuring \DQC$k$ machines: 
\begin{corollary}[{\corref{cor:meas}, stated informally}]\label{cor:meas''}
    If $n>k$ then it is impossible to synthesize a $n$-qubit pure state on a discarding \DQC$k$ machine with probability greater than $2^{k-n}$. 
\end{corollary}
We also highlight \thmref{thm:tr''}'s applicability to current questions of algorithm design. Recent work such as \cite{CDL,MS} (part of a trend towards Monte Carlo or Lindbladian approaches attempting to simulate nature; see also e.g.\ \cite{TOV+,YA,CB,CBKG,Cub,KBC+,CL,CW,LW,DLL,RFA}) has studied the iterated problem with $a$ ancilla qubits in $\fA$ (often $a=1$) always initialized to a pure state and an aribitrary starting state in $\fB$. Such a circuit is given in \figref{fig:cdl}. 
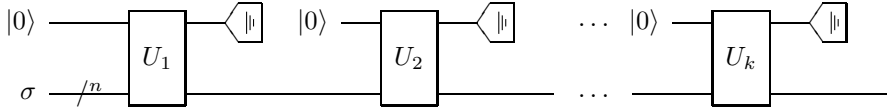
\begin{figure}
    \centering
    \leavevmode
    \Qcircuit @C=1.5em @R=1.5em {
    \lstick{\ket{0}} & \qw & \multigate{1}{U_1} \qw & \measuretab{\Ground{}} & & \lstick{\ket{0}} & \multigate{1}{U_2} \qw & \measuretab{\Ground{}} & & \dots & & \lstick{\ket{0}} & \multigate{1}{U_k} \qw & \measuretab{\Ground{}} \\
    \lstick{\gs} & {/^n} \qw & \ghost{U_1} & \qw & \qw & \qw & \ghost{U_2} & \qw & \qw & \dots & & \qw & \ghost{U_k} & \qw & \qw 
    }
    \captionsetup{width=.75\textwidth}
    \caption{A circuit of interest in algorithms such as \cite{CDL,MS}.}
    \label{fig:cdl}
\end{figure}
One intriguing result of \cite{CDL} is that it is always possible to pick each circuit to be local and arrange for convergence of the output state to the ground state of a given Hamiltonian, regardless of the initial state on $\fB$, even if that state is orthogonal to the target ground state. We are here able to establish lower bounds on that rate of convergence, in the instance that the second register begins maximally mixed, a sort of ``worst case'' input (intuitively speaking). 
\begin{corollary}[{\corref{cor:cdl}, stated informally}]\label{cor:cdl''}
    It is impossible to synthesize a $n$-qubit pure state on a \DQC1 machine with refresh in fewer than $n$ iterations. 
\end{corollary}

The tightness of this lower bound is rather straightforward and is independent of the input state on $\fB$. It is depicted in \figref{fig:cdl''}. 
\begin{figure}
    \centering
    \leavevmode
    \Qcircuit @C=1.5em @R=1.5em {
    \lstick{\ket{0}} & \qw & \qswap \qwx[1] \qw & \measuretab{\Ground{}} & & \lstick{\ket{0}} & \qswap \qwx[1] \qw & \measuretab{\Ground{}} & & \cdots & & \lstick{\ket{0}} & \qswap \qwx[1] \qw & \measuretab{\Ground{}} \\
    \lstick{\frac{1}{2^n}\id} & {/^n} \qw & \gate{\times_1} & \qw & \qw & \qw & \gate{\times_2} & \qw & \qw & \cdots & & \qw & \gate{\times_n} & \qw & \qw 
    }
    \captionsetup{width=.75\textwidth}
    \caption{A circuit witnessing the lower bound in \corref{cor:cdl''}. (The gate $\times_i$ performs a swap on only the $i$th wire in the register.)}
    \label{fig:cdl''}
\end{figure}
We remark that there is a good deal known already about the hardness of ground state preparation \cite{AN,KSV,KKR,AGIK} and indeed that it is \QMA-hard. There is also work with different models models \cite{GR,Aar,INN+}, culminating in the recent paper \cite{Ros} which gives a one-shot algorithm but measures algorithmic cost differently. 

We may then consider the problem of {\em approximate state synthesis}: what if the input state to the \DQC$k$ machine (measuring or discarding) is only promised to be accurate to precision $\eps$, in some suitable sense, and the target need only be returned to precision $\eps'$? (We may let either $\eps$ or $\eps'$ be 0. See \secref{sec:noise} for the precise definition of accuracy used.) For now, let such machines tolerating $\eps$ imprecision on the input be {\em measuring/discarding $\eps$-close \DQC$k$ machines}. We show analogues to \thmref{thm:tr''} and \corref{cor:meas''} for this robust setting, which also allows \corref{cor:cdl''} to be relaxed to a $\frac{1}{4}$-close state to the pure target state. 
\begin{proposition}[{\propref{prop:error}, stated informally}]\label{prop:error''}
    If $n>k$, then: 
    \begin{itemize}
        \item if $\eps+\eps'\le\frac{1}{4}$ it is impossible to synthesize a state $\eps'$-close to a pure state on a discarding $\eps$-close \DQC$k$ machine. 
        \item it is impossible to synthesize a state $\eps'$-close to a pute state on a measuring $\eps$-close \DQC$k$ machine with probability greater than 
        $$\frac{2^{k-n}+2\eps}{1-2\eps'}.$$
    \end{itemize}
\end{proposition}

Finally, we study the special case of {\em Gibbs state synthesis}: what if the target state to the \DQC$k$ machine is a Gibbs state of a given Hamiltonian at given temperature? We further show analogues to \thmref{thm:tr''} and \corref{cor:meas''} for the Gibbs state setting: 
\begin{corollary}[{\corref{cor:gibbs}, stated informally}]\label{cor:gibbs''}
    If $n>k$ and a Hamiltonian on the output register has spectral gap $\g$ and the inverse temperature is at least the order of $\frac{n}{\g}$, then it is impossible to synthesize the Hamiltonian's Gibbs state at the given temperature:
    \begin{itemize}
        \item on a discarding \DQC$k$ machine. 
        \item on a measuring \DQC$k$ machine, with probability greater than $2^{1+k-n}$. 
    \end{itemize}
\end{corollary}
We also point out that for clarity's sake we have focused on the case of the target state being rank-one, however one could readily adapt the methods herein to more general low-rank mixtures. 

Before delving into the proofs, we lastly remark that the results here are relatively elementary---they boil down to manipulating matrices on the level of entries---but having not seen such results in the literature before, we present them in case they are of interest or of use. 

The structure of the paper is as follows. In \secref{sec:tr} we prove \thmref{thm:tr''} about discarding \DQC$k$ machines, in \secref{sec:meas} we prove \corref{cor:meas''} about measuring \DQC$k$ machines, in \secref{sec:noise} we prove \propref{prop:error''} about noise on the input and/or output, in \secref{sec:gibbs} we prove \corref{cor:gibbs''} about prepration of Gibbs states at low temperatures, and in \secref{sec:cdl} we prove \corref{cor:cdl''} about an application to current algorithms. In \secref{sec:future} we conclude and discuss possible future avenues of research.

\section{No-go result for tracing out \texorpdfstring{(\thmref{thm:tr''})}{}}\label{sec:tr}

\begin{theorem}[{\thmref{thm:tr''}, restated formally}]\label{thm:tr}
    Given input state $\rho=\ket{0}\bra{0}_\fA\otimes\frac{1}{2^b}\id_\fB$ on the registers $\fA\sqcup\fB$ and a subset $\fC$ of the qubits with complement $\fD$ having $d>a$ qubits, there does not exist a circuit $U$ such that tracing out $\fC$ leaves $U\rho U^\dag$ in a pure state on $\fD$. (See \figref{fig:1}.)
\end{theorem}

The quantum channels of interest are of the form $\tr_\fC U[\rho]$, where $\rho=\ket{0}\bra{0}_\fA\otimes\frac{1}{2^b}\id_\fB$ is the input state; $\tr_\fC$ denotes the partial trace from discarding the qubits in register $\fC$; $U[\gs]=U\gs U^\dag$ is conjugation; and $U$ is the circuit to be found. Thus we say that 
$$\cE_{U,\fC}(\gs)=\tr_\fC U[\gs]$$
is the relevant quantum channel. We claim that $\cE_{U,\fC}(\rho)\neq\ket{0}\bra{0}_\fD$ where $\fD$ is $\fC$'s complement. 

\begin{remark}\label{rmk:any state}
    It is important to note that $\ket{0}\bra{0}_\fD$ could be any other pure state: suppose instead we had as our target $\ket{\psi}\bra{\psi}_\fD$, and there were a unitary $U$ such that $\cE_{U,\fC}(\rho)=\ket{\psi}\bra{\psi}_\fD$. Let $V_\fD$ be any circuit sending $\ket{\psi}_\fD$ to $\ket{0}_\fD$. Then $\cE_{(\id_\fC\otimes V_\fD)U,\fC}(\rho)=\ket{0}\bra{0}_\fD$. Thus the problems for target state $\ket{0}_\fD$ and $\ket{\psi}_\fD$ are equivalent, and we lose nothing by considering $\ket{0}\bra{0}_\fD$ without loss of generality. 
\end{remark}

We introduce some further notation and shorthand. Continue to say that $a=\#\fA$, $b=\#\fB$, $c=\#\fC$, and $d=\#\fD$, satisfying $n=a+b=c+d$, and write $A=2^a$, $B=2^b$, $C=2^c$, $D=2^d$, and $N=2^n$. 

Without loss of generality we may take $\fA$ to be the top $a$ wires and $\fD$ to be the bottom $d$ wires, after applying all relevant $\SWAP$ gates. Say also that $\db{N}=\{0,\dots,N-1\}$; by default, index density matrices on the full register with $\db{N}\times\db{N}$, in correspondence with the computational basis.

\begin{proof}[{Proof of \thmref{thm:tr}}]
Let $\gs$ be any density matrix on the full set of qubits $\fA\sqcup\fB=\fC\sqcup\fD$. Then,
\begin{equation}
    \tr_\fD\gs=\sum\limits_{j=0}^{C-1}(\bra{j}_\fC\otimes\id_\fD)\gs(\ket{j}_\fC\otimes\id_\fD). \label{eq:trace out}
\end{equation}
We can write $\ket{j}_\fC\otimes\id_\fD$ as the column vector of square matrices which is 0 except in the $j$th block where it is $\id_\fD$, i.e.\ $\begin{pmatrix}
    0_\fD & \dots & 0_\fD & \id_\fD & 0_\fD & \dots & 0_\fD
\end{pmatrix}^\dag$, so overall in the left column the 1 appears in position $jD$. 

Take some candidate $U$ and label its columns $\bu_0,\dots,\bu_{N-1}$. Then, compute $U(\ket{j}_\fC\otimes\id_\fD)$ to be the columns $\bu_{jD},\dots,\bu_{(j+1)D-1}$ of $U$. 

The input state $\rho$ is $\frac{1}{B}$ times the square $N\times N$ block-diagonal matrix where the top-left block is a $B$-dimensional identity matrix and all other entries are 0. Notice that the block-diagonal matrix is real, Hermitian, and its own square. Call it $P$. Then, we wish to compute $P_j=PU(\ket{j}_\fC\otimes\id_\fD)$ and then $\frac{1}{B}P_j^\dag P_j$. We compute that $P_j$ is the $N\times D$ matrix where for $x\in\db{B}$ and $y\in\db{D}$, $(P_j)_{x,y}=u_{jD+y,x}$, and for rows $B$ and below $P_j$ is all 0. 

We then compute that
$$\lpr{\frac{1}{B}P_j^\dag P_j}_{0,0}=\frac{1}{B}\sum\limits_{x=0}^{B-1}\abs{u_{jD,x}}^2.$$
Thus 
\begin{equation}
    1=\lpr{\cE_{U,\fC}(\rho)}_{0,0}=\frac{1}{B}\sum\limits_{j=0}^{C-1}\sum\limits_{x=0}^{B-1}\abs{u_{jD,x}}^2\le\frac{C}{B}\label{eq:punchline}
\end{equation}
(the first equality following by the definition of any state $\ket{0}\bra{0}$, the second by \eqref{eq:trace out}, and the inequality since each of the $C$ sums is a partial column norm of a unitary matrix), so that $d\le a$, a contradiction. Consequently, there does not exist a circuit $U$ such that tracing out $\fC$ leaves $U\rho U^\dag$ in a pure state on $\fD$. 
\end{proof}

\section{A bound on the probability when measuring \texorpdfstring{(\corref{cor:meas''})}{}}\label{sec:meas}
 
The following gives a bound when measuring $\fC$, where the probability of successfully obtaining the target pure state is upper-bounded by a quantity decaying exponentially in the number of additional qubits to be synthesized. 

\begin{corollary}[{\corref{cor:meas''}, restated formally}]\label{cor:meas}
    Given input state $\rho=\ket{0}\bra{0}_\fA\otimes\frac{1}{2^b}\id_\fB$ on the registers $\fA\sqcup\fB$ and a subset $\fC$ of the qubits with complement $\fD$ having $d>a$ qubits, for all circuits $U$, the probability that measuring $\fC$ leaves $U\rho U^\dag$ in a pure state on $\fD$ is at most $2^{a-d}$. 
\end{corollary}

In particular, if $d=a+\gw(\poly\log a)$ then the probability is $o\lpr{\frac{1}{\poly a}}$, so that the probability of synthesizing logarithmically-many additional qubits is of relatively small order; when $a=\gT(\log n)$ (recall from \cite{Shepherd} that this does not change the class of decision problems solvable by the model), this probability is $o\lpr{\frac{1}{\poly\log n}}$.

Towards the proof, we take the setup from \thmref{thm:tr} with a slightly different quantum channel $\cE'_{U,\fC}(\rho)$, where the state $U[\rho]$ is measured on the qubits in $\fC$. Then, we study $\bP\lbr{\cE'_{U,\fC}(\rho)=\ket{0}\bra{0}_\fD}$. (We maintain the target state of $\ket{0}_\fD$ for the exact same reason as given in \rmkref{rmk:any state}.)

\begin{proof}
    We reuse some of the work from the preceding proof. In the setting of \eqref{eq:trace out}, call 
    \begin{align*}
        T_j(\gs)=(\bra{j}_\fC\otimes\id_\fD)\gs(\ket{j}_\fC\otimes\id_\fD), && p_j(\gs)=\tr T_j(\gs), && S_j(\gs)=\frac{T_j(\gs)}{p_j(\gs)}.
    \end{align*}
    Here, $p_j$ represents the probability of measuring $j$, and $S_j$ is the resulting state on $\fD$, so that $\tr_\fD\gs=\sum\limits_{j=0}^{C-1}S_j(\gs)p_j(\gs)$. $S_j$ is positive semidefinite, hence its $(0,0)$ entry is nonnegative, and of course $p_j\ge0$ because it is a probability. Consider again $\gs=\rho$. Let $\cO\subset\db{C}$ be the set of measurements on $\fC$ which leave the system in $\ket{0}\bra{0}_\fD$; if $j\in\cO$ then $S_j(\rho)_{0,0}=1$, and so $\bP\lbr{\cE'_{U,\fC}(\rho)=\ket{0}\bra{0}_\fD}=\sum\limits_{j\in\cO}p_j(\rho)=\sum\limits_{j\in\cO}S_j(\rho)_{0,0}p_j(\rho)$, thus
    \begin{align}
        \bP\lbr{\cE'_{U,\fC}(\rho)=\ket{0}\bra{0}_\fD}\le\sum\limits_{j=0}^{C-1}S_j(\rho)_{0,0}p_j(\rho)=\cE_{U,\fC}(\rho)_{0,0}\le\frac{C}{B}=2^{a-d}. \label{eq:E' to E}
    \end{align}
    The first equality recalls $\cE_{U,\fC}$ from \thmref{thm:tr}, the second inequality invokes \eqref{eq:punchline}, and the remaining steps are evident. 
\end{proof}
This bound can be saturated by letting $U$ be $\SWAP$ gates which ensure that the output register contains the $a$ pure input qubits. 

\section{Robustness under noise \texorpdfstring{(\propref{prop:error''})}{}}\label{sec:noise}
We now consider the introduction of noise. The most natural metric to use on the set of density matrices is the trace distance $\d(\gs,\gs')=\half\tr\abs{\gs-\gs'}=\half\norm{\gs-\gs'}_1$, though we will often use the weaker operator norm, via the (dimension-independent) inequality $\norm{\gs-\gs'}\le2\d(\gs,\gs')$. Accordingly, we say that $\gs$ is {\em $\eps$-close} to $\gs'$ if $\d(\gs,\gs')\le\eps$. 

Then, we show that \thmref{thm:tr} and \corref{cor:meas} are robust, in the sense that the former persists if $\eps+\eps'\le\frac{1}{4}$, and the latter where the probability is upper-bounded by $\frac{2^{a-d}+2\eps}{1-2\eps'}$. 

The former claim (\propref{prop:error}(a)) is especially strong: we may allow $\eps$ and $\eps'$ to both be constant-sized, say $\frac{1}{8}$. This amounts to a ``no-go zone'' in the space of input density matrices which is a constant-radius $\ell^1$ ball about $\rho$. The second is still of interest, but especially so when $\eps=o\lpr{\frac{1}{\poly\log n}}$ and $\eps'<\half$ so that the bound on the probability from the \secref{sec:meas} is unchanged up to prefactor. 

\begin{proposition}[{\propref{prop:error''}, restated formally}]\label{prop:error}
    Given input state $\tilde{\rho}$ which is $\eps$-close to $\ket{0}\bra{0}_\fA\otimes\frac{1}{2^b}\id_\fB$ on the registers $\fA\sqcup\fB$ and a subset $\fC$ of the qubits with complement $\fD$, as well as a second precision parameter $\eps'$:
    \begin{enumerate}
        \item if $\eps+\eps'\le\half-2^{a-d-1}$, there does not exist a circuit $U$ such that tracing out $\fC$ leaves $U\tilde{\rho}U^\dag$ in a state which is $\eps'$-close to a pure state on $\fD$. 
        \item for all circuits $U$, the probability that measuring $\fC$ leaves $U\tilde{\rho}U^\dag$ in a state which is $\eps'$-close to a pure state on $\fD$ is at most $\frac{2^{a-d}+2\eps}{1-2\eps'}$. 
    \end{enumerate}
\end{proposition}

We first prove a useful lemma. 
\begin{lemma}\label{lem:00 error}
    Suppose $\gs$ and $\gs'$ are two density matrices on the same set of qubits $\fE$ with $\d(\rho,\gs)<\eps$. Then for any $\fF\subset\fE$, $\abs{\lpr{\tr_\fF(\gs-\gs')}_{0,0}}<2\eps$. 
\end{lemma}
\begin{proof}
    Let $\fF$ be of size $f$, $\fG=\fE\backslash\fF$ be of size $g$, $\gd=\gs-\gs'$, without loss of generality $\fF$ is the set of the bottom $f$ qubits in $\fE$, and suppose the computational basis on $\fF$ is enumerated on $\db{F}$. We recall \eqref{eq:trace out}, and notice that
    $$\lpr{\tr_\fF(\gs-\gs')}_{0,0}=\bra{0}_\fF\tr_\fF(\gs-\gs')\ket{0}_\fF=\lpr{1\otimes\bra{0}_\fF}\lpr{\sum\limits_{j=0}^{F-1}\lpr{\bra{j}_\fG\otimes\id_\fF}\gd\lpr{\ket{j}_\fG\otimes\id_\fF}}\lpr{1\otimes\ket{0}_\fF}.$$
    This in turn equals
    $$\sum\limits_{j=0}^{F-1}\bra{j,0}\gd\ket{j,0}=\tr\sum\limits_{j=0}^{F-1}\bra{j,0}\gd\ket{j,0}=\tr\lpr{\gd\cdot\id_\fG\otimes\ket{0}\bra{0}_\fF}=\tr\lpr{(PU)^\dag\gL(PU)}$$
    for $P=\id_\fG\otimes\ket{0}\bra{0}_\fF$ (a projector, so its own square) and the unitary diagonalization $\gd=U^\dag\gL U$ where $\gL=\diag\bgl$ for $\bgl=(\gl_0,\dots,\gl_{N-1})$ satisfying $\norm{\bgl}_1<2\eps$. Write also $U=(u_{j,k})_{j,k\in\db{N}}$; clearly $\abs{u_{j,k}}\le1$ by unitarity. Then we readily compute that
    \begin{align*}
        \lpr{\tr_\fF(\gs-\gs')}_{0,0}=\sum\limits_{j=0}^{G-1}\abs{u_{jF,jF}}^2\gl_{jF}, &&\text{hence}&&
        \abs{\lpr{\tr_\fF(\gs-\gs')}_{0,0}}\le\sum\limits_{j=0}^{G-1}\abs{\gl_{jF}}\le\norm{\bgl}_1<2\eps
    \end{align*}
    and so we are done. 
\end{proof}

\begin{proof}[{Proof of \propref{prop:error}}]
    \begin{enumerate}
        \item We use linearity of $\cE_{U,\fC}$. Let $\tilde{\rho}$ be the $\eps$-close input state, so that $\d(\rho,\tilde{\rho})<\eps$ hence $\tilde{\rho}=\rho+\gd$ where $\norm{\gd}_1<2\eps$. Then, $\cE_{U,\fC}(\tilde{\rho})=\cE_{U,\fC}(\rho)+\cE_{U,\fC}(\gd)$. By unitary invariance of $\d$, i.e.\ $\d(U[\gs],U[\gs'])=\d(\gs,\gs')$ for all density matrices $\gs$ and $\gs'$, and linearity of $U[\cdot]$, we apply \lemref{lem:00 error} and conclude that $\cE_{U,\fC}(\tilde{\rho})<\frac{C}{B}+2\eps$. If this entry is smaller than $1-2\eps'$ then $\cE_{U,\fC}(\tilde{\rho})$ cannot be $\eps'$-close to $\ket{0}\bra{0}_\fD$. Thus it suffices to have $\eps+\eps'<\half-2^{a-d-1}$, with $2^{a-d-1}\le\frac{1}{4}$ always. 
        \item Apply the same reasoning to \eqref{eq:E' to E}: let $\cO\subset\db{C}$ be the set of measurements on $\fC$ which leave the system $\eps'$-close to $\ket{0}\bra{0}_\fD$; $\cO\subset\cO'$ which is the set of measurements for which $S_j\lpr{\tilde{\rho}}>1-2\eps'$. So, 
        $$\bP\lbr{\text{$\cE'_{U,\fC}\lpr{\tilde{\rho}}$ is $\eps$-close to $\ket{0}\bra{0}_\fD$}}\le\sum\limits_{j\in\cO}p_j\lpr{\tilde{\rho}}\le\sum\limits_{j\in\cO'}p_j\lpr{\tilde{\rho}}\le\frac{1}{1-2\eps'}\sum\limits_{j\in\cO'}S_j\lpr{\tilde{\rho}}_{0,0}p_j\lpr{\tilde{\rho}}.$$
        Using more bounds by expanding the index of a positive sum, this in turn is bounded by
        $$\frac{1}{1-2\eps'}\sum\limits_{j=0}^{D-1}S_j\lpr{\tilde{\rho}}_{0,0}p_j\lpr{\tilde{\rho}}=\frac{\cE_{U,\fC}\lpr{\tilde{\rho}}_{0,0}}{1-2\eps'}\le\frac{2^{a-d}+2\eps}{1-2\eps'}$$
        where in the last step we use the previous part of this proof instead of \eqref{eq:punchline}. 
    \end{enumerate}
\end{proof}

\section{Low-temperature Gibbs states \texorpdfstring{(\corref{cor:gibbs''})}{}}\label{sec:gibbs}

We consider now the task of preparing the Gibbs state for a Hamiltonian on a subsystem. Recall that for a Hamiltonian $H$ at inverse temperature $\gb$, the {\em Gibbs state} is 
$$\tG_\gb(H)=\frac{\exp(-\gb H)}{\tr\exp(-\gb H)}$$
(where $\tZ_\gb(H)=\tr\exp(-\gb H)$ is the {\em partition function}). \propref{prop:error} enables us to show a no-go result for Gibbs states of sufficiently low temperature for gapped Hamiltonians. 

We say that a Hamiltonian $H$ is {\em $\g$-gapped} if it has a unique least eigenvalue and all other eigenvalues differ by at least $\g$ from it. Shift $H$ as necessary to have least eigenvalue equal to 0, and notice that this does not change $\tG_\gb$. Then, let $H$'s eigenpairs be $\lpr{\gl,\ket{\psi_\gl}}$. It follows that $\ket{\psi_0}$ is $H$'s ground state. 

The core idea is that $\tG_\infty(H)=\ket{\psi_0}\bra{\psi_0}_\fD$, so $\tG_\gb(H)$ for $\gb\gg1$ must be very close to this pure state (by continuity), hence by \rmkref{rmk:any state} we can use the results of \secref{sec:noise}. 

We first prove a useful lemma. In terms of any $\eps'$ (and the system parameters $a$, $b$, $c$, $d$, and $\g$), we work to establish which $\gb$ enforce $\d\lpr{\tG_\gb(H_\fD),\ket{\psi_0}\bra{\psi_0}_\fD}<\eps'$. 
\begin{lemma}\label{lem:gibbs lemm}
    If $\gb>\frac{d\ln2+\ln\frac{1-\eps'}{\eps'}}{\g}$ then $\d\lpr{\tG_\gb(H_\fD),\ket{\psi_0}\bra{\psi_0}_\fD}<\eps'$.
\end{lemma}
\begin{proof}
    We compute 
    \begin{align*}
        \tG_\gb(H_\fD)-\ket{\psi_0}\bra{\psi_0}_\fD&=\lpr{\frac{1}{\tZ_\gb(H_\fD)}-1}\ket{\psi_0}\bra{\psi_0}_\fD+\sum\limits_{\gl\neq0}\frac{\exp(-\gb\gl)}{\tZ_\gb(H_\fD)}\ket{\psi_\gl}\bra{\psi_\gl}_\fD\\
        \norm{\tG_\gb(H_\fD)-\ket{\psi_0}\bra{\psi_0}_\fD}_1&=1-\frac{1}{\tZ_\gb(H_\fD)}+\sum\limits_{\gl\neq0}\frac{\exp(-\gb\gl)}{\tZ_\gb(H_\fD)}
            =2-\frac{2}{\tZ_\gb(H_\fD)}.
    \end{align*}
    We are then concerned with enforcing $1-\frac{1}{\tZ_\gb(H_\fD)}<\eps'$, i.e.\ $\tZ_\gb(H_\fD)<\frac{1}{1-\eps'}$. 
    We notice that $\tZ_\gb(H_\fD)\le1+\exp(-\gb\g)(D-1)$ (this is tight absent further spectral information about $H_\fD$) and so $1+\exp(-\gb\g)(D-1)<\frac{1}{1-\eps'}$ is sufficient. We then merely solve for $\gb$. 
\end{proof}
We now analogize \thmref{thm:tr} for a gapped Hamiltonian at low temperature, and give a bound \`a la \corref{cor:meas} where the probability is upper-bounded by a quantity decaying exponentially in the number of additional qubits to be used for the Gibbs state. 
\begin{corollary}[{\corref{cor:gibbs''}, restated formally}]\label{cor:gibbs}
    Given input state $\rho=\ket{0}\bra{0}_\fA\otimes\frac{1}{2^b}\id_\fB$ on the registers $\fA\sqcup\fB$ and a subset $\fC$ of the qubits with complement $\fD$, as well as a $\g$-gapped Hamiltonian $H_\fD$ on $\fD$ and an inverse temperature parameter $\gb$, if $\gb>\frac{d\ln2+\ln3}{\g}$:
    \begin{enumerate}
        \item There does not exist a circuit $U$ such that tracing out $\fC$ leaves $U\tilde{\rho}U^\dag$ in the state $\tG_\gb(H_\fD)$. 
        \item For all circuits $U$, the probability that measuring $\fC$ leaves $U\tilde{\rho}U^\dag$ in the state $\tG_\gb(H_\fD)$ is at most $2^{a-d+1}$. 
    \end{enumerate}
\end{corollary}
\begin{proof}
    Since we are starting 0-close to $\ket{0}\bra{0}_\fA\otimes\frac{1}{2^b}\id_\fB$, we let $(\eps,\eps')=\lpr{0,\frac{1}{4}}$ and apply \lemref{lem:gibbs lemm} and \propref{prop:error}. 
\end{proof}

\section{An algorithmic application \texorpdfstring{(\corref{cor:cdl''})}{}}\label{sec:cdl}

Recent work including \cite{CDL,MS} has studied the task of preparing various quantum states using iterated single-ancilla algorithms, or generally repeated interaction algorithms, where at each iteration the ancilla qubits are discarded and replaced with a fresh pure qubit. 
Typically, the target state is the ground state of a prescribed Hamiltonian, and the starting state is arbitrary. Indeed, \cite{CDL} attains an algorithm using Lindbladian evolution with a single ancilla with refresh, which may even start from a state which is orthogonal to the target. Using the techniques developed earlier (\thmref{thm:tr}) we give a lower bound on the required number of single-qubit interactions required to prepare a target ground state or any low-temperature target Gibbs state of a system Hamiltonian, when the starting state is maximally-mixed (i.e., without any prior knowledge about the overlap of the initial state and ground state). The lower bound is linear in system size and independent of the Hamiltonian, and is an example of an information-theoretic lower bound on computational complexity of such an algorithm.

\begin{corollary}[{\corref{cor:cdl''}, restated formally}]\label{cor:cdl}
    Consider an algorithm that takes in a single pure ancilla $\ket{0}$ and any $n$-qubit starting state $\gs$. Each step consists of evolving the entire system by some circuit and discarding the ancilla qubit. The result is then given a fresh ancilla qubit $\ket{0}$. After $k$ steps, the result of the discarding is the algorithm's output. If $\gs$ is the maximally-mixed state on $n$ qubits, then this algorithm requires at least $k\ge n$ steps to reach any pure state. If instead there are $a$ pure ancillae then the number of steps required is $k\ge\frac{n}{a}$. 
\end{corollary}
Notice that the result is completely agnostic to the structure of the circuits chosen at each step. This immediately gives a lower bound on time and query complexity of the ground state and Gibbs state synthesis frameworks considered in \cite{CDL,MS}.
\begin{proof}
    We simply rewrite \figref{fig:cdl} as \figref{fig:moved}. 
\begin{figure}
    \centering
    \leavevmode
    \Qcircuit @C=1.5em @R=1.5em {
    \lstick{\ket{0}} & \qw & \gate{U_1} & \qw & \qw & \dots & & \qw & \measuretab{\Ground{}} \\
    \lstick{\ket{0}} & \qw & \qw \qwx & \gate{U_2} & \qw & \dots & & \qw & \measuretab{\Ground{}} \\
    \vdots & & \qwx & \qwx \\
    \lstick{\ket{0}} & \qw & \qw \qwx & \qw \qwx & \qw & \dots & & \multigate{1}{U_k} & \measuretab{\Ground{}} \\
    \lstick{\gs} & {/^n} \qw & \gate{U_1} \qwx & \gate{U_2} \qwx & \qw & \dots & & \ghost{U_k} & \qw & \qw
    }
    \captionsetup{width=.75\textwidth}
    \caption{The circuit from \figref{fig:cdl} ``unfolded'' to match the format of Figures \ref{fig:1} and \ref{fig:2}. Here, gates such as $U_1$ and $U_2$ are depicted as ``split'' but are just the corresponding gates from \figref{fig:cdl}, which act specifically on, in the case of $U_1$, the first wire and the last register. }
    \label{fig:moved}
\end{figure}
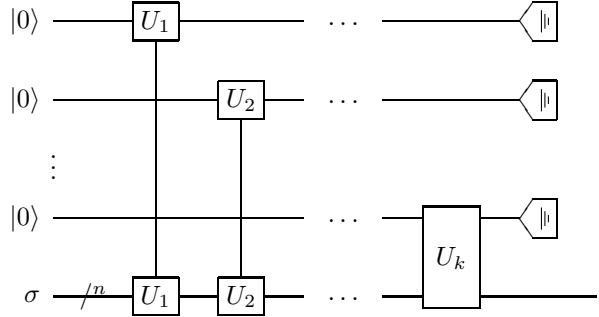
    Note that up to relabeling circuit components this is exactly \figref{fig:1}. The result follows from \thmref{thm:tr}. 
\end{proof}
This bound is evidently saturated in the $a=1$ case by the circuit in \figref{fig:cdl''}. Invoking \propref{prop:error}(a) instead of \thmref{thm:tr} allows the result to be loosened to an output state which is merely $\frac{1}{4}$-close to a pure (or low-temperature Gibbs) state. 

\section{Conclusions and future directions}\label{sec:future}

In this work, we study the problem of pure and Gibbs state synthesis in the One Clean Qubit model and its generalization, with $a$ clean qubits. By studying the action of a circuit purported to prepare such states (after measuring or tracing out), we show that it is impossible (in the case of measuring, with high probability) to prepare the desired state on more than $a$ qubits (and low tempreature, in the case of Gibbs states). This is then applied to the setting of iterated algorithms in the presence of a stream of single pure qubits to show that algorithms of this form require at least (indeed, exactly) $n$ rounds to reach a pure state on $n$ qubits. 

I am grateful to Subhayan Roy Moulik \cite{SRM} for pointing out the elegant alternative proof of \thmref{thm:tr} that follows. 

\begin{proof}[{Entropic proof of \thmref{thm:tr}}]
    Consider the von Neumann entropy of a density matrix $\gs$ on $n$ qubits: $\tH(\gs)=-\tr(\gs\log\gs)$. $\tH$ is additive across tensor factors and $\tH(\gs)=0$ if $\gs$ is a pure state and in general $H(\gs)\le n$. Another immediate property of $\tH$ is invariant under unitary conjugation: $\tH(\gs)=\tH(U[\gs])$. Moreover, due to the Araki--Lieb inequality \cite{AL}, for complementary subsystems of qubits, say $\fC$ and $\fD$, $\tH(\gs)\le \tH\lpr{\tr_\fC\gs}+\tH\lpr{\tr_\fD\gs}$. Let $\rho = \ket{0}\bra{0}_\fA \otimes \frac{1}{2^b}\id_\fB$; then $\tH(\rho)=b$. Putting this all together, for any $U$ and $\fC$ (with its complement $\fD$), when $\tr_\fC U[\rho]$ is a pure state, with zero entropy, we then have
    $$b=\tH(\rho)=\tH(U[\rho])\le\tH\lpr{\tr_\fC U[\rho]}+\tH\lpr{\tr_\fD U[\rho]}=\tH\lpr{\tr_\fD U[\rho]}\le c,$$
    and since $n=a+b=c+d$, it follows $d\le a$. 
\end{proof}

One interesting direction is to find entropic/information-theoretic proofs of \corref{cor:meas}, \propref{prop:error}, and \corref{cor:gibbs}; as matters stand presently, the linear algebra techniques are rather ad hoc and a more systematic approach would be satisfying. 

Another avenue is to apply these techniques to other algorithms for pure state preparation and possibly obtain comparable lower bounds, while we search for settings admitting algorithms running in a sublinear number of rounds. 

One more direction we propose is that these techniques may find use in showing lower bounds for state complexity, the number of basic circuit components required to synthesize a given pure state from a fixed fresh register. 

Finally, we note that as we may deal with state synthesis as well as decision problems in these frameworks, we may explore comparisons between \P, \DQC1 (the complexity class), and \BQP. Perhaps work along the lines presented here is the start to a means by which we can obtain separation between the former and latter of these complexity classes. 

\section*{Acknowledgements}

This work was supported by a NSF graduate research fellowship, grant number DGE 2146752. I am grateful to Subhayan Roy Moulik for many conversations which led to the production of this work as well as for pointing out many interesting questions and features. I am also grateful to Nikhil Srivastava and Matt Tyler for helpful feedback.


\begin{thebibliography}{}
{\footnotesize
\bibitem[Aar'09]{Aar} S.\ Aaronson, ``Quantum copy-protection and quantum money,'' {\em IEEE Conference on Computational Complexity}, 2009. 
\bibitem[AGIK'09]{AGIK} D.\ Aharonov, D.\ Gottesman, S.\ Irani, and J.\ Kempe, ``The power of quantum systems on a line,'' {\em Communications in Mathematical Physics} 287(1), pp.\ 41--65, 2009.
\bibitem[AN'02]{AN} D. Aharonov and T. Naveh, ``Quantum \NP\ -- A Survey,'' \arxiv{quant-ph/0210077}, 2002. 
\bibitem[ADLH'05]{ADLH} A.\ Aspuru-Guzik, A.\ D.\ Dutoi, P.\ J.\ Love, and M.\ Head-Gordon, ``Simulated quantum computation of molecular energies,'' {\em Science} 309(5741), pp.\ 1704--1707, 2005.
\bibitem[AL'70]{AL} H.\ Araki and E.\ H.\ Lieb, ``Entropy Inequalities,'' {\em Communications in Mathematical Physics} 18, pp.\ 160--170, 1970. 
\bibitem[CB'23]{CB} C.-F.\ Chen and F.\ G.\ S.\ L.\ Brand\~ao, ``Fast thermalization from the eigenstate thermalization hypothesis,'' \arxiv{2112.07646}, 2023.
\bibitem[CKBG'23]{CBKG} C.-F.\ Chen, M.\ J.\ Kastoryano, F.\ G.\ S.\ L.\ Brand\~ao, and A.\ Gily\'en, ``Quantum Thermal State Preparation,'' \arxiv{2303.18224}, 2023. 
\bibitem[CL'17]{CL} A.\ M.\ Childs and T.\ Li, ``Efficient simulation of sparse Markovian quantum dynamics,'' {\em Quantum Information and Computation} 17(11\&12), pp.\ 901--947, 2017. 
\bibitem[CW'17]{CW} R.\ Cleve and C.\ Wang, ``Efficient Quantum Algorithms for Simulating Lindblad Evolution,'' {\em International Colloquium on Automata, Languages, and Programming}, 2017. 
\bibitem[Cub'23]{Cub} T.\ S.\ Cubitt, ``Dissipative ground state preparation and the dissipative quantum eigensolver,'' \arxiv{2303.11962}, 2023. 
\bibitem[DCL'23]{CDL} Z.\ Ding, C.-F.\ Chen, and L.\ Lin, ``Single-ancilla ground state preparation via Lindbladians,'' \arxiv{2308.15676}, 2023. 
\bibitem[DLL'24]{DLL} Z.\ Ding, X.\ Li, and L.\ Lin, ``Simulating open quantum systems using Hamiltonian simulations,'' {\em PRX Quantum}, 2024. 
\bibitem[GTC'19]{GTC} Y.\ Ge, J.\ Tura, and J.\ I.\ Cirac, ``Faster ground state preparation and high-precision ground energy estimation with fewer qubits,'' {\em Journal of Mathematical Physics} 60(2), 2019.
\bibitem[GR'02]{GR} L.\ Grover and T.\ Rudolph, ``Creating superpositions that correspond to efficiently integrable probability distributions,'' \arxiv{quant-ph/0208112}, 2002. 
\bibitem[INN$^+$'22]{INN+} S.\ Irani, A.\ Natarajan, C.\ Nirkhe, S.\ Rao, and H.\ Yuen, ``Quantum search-to-decision reductions and the state synthesis problem,'' {\em IEEE Conference on Computational Complexity}, 2022. 
\bibitem[KKR'06]{KKR} J.\ Kempe, A.\ Kitaev, and O.\ Regev, ``The complexity of the local Hamiltonian problem,'' {\em SIAM Journal on Computing} 35(5), pp.\ 1070--1097, 2006.
\bibitem[KSV'02]{KSV} A.\ Y.\ Kitaev, A.\ Shen, and M.\ N.\ Vyalyi, ``Classical and quantum computation,'' {\em Graduate Studies in Mathematics}, American Mathematical Society, 2002. 
\bibitem[KBC$^+$'11]{KBC+} M.\ Kliesch, T.\ Barthel, C.\ Gogolin, M.\ Kastoryano, and J.\ Eisert, ``Dissipative Quantum Church-Turing Theorem,'' {\em Physical Review Letters} 107, 2011. 
\bibitem[KL'98]{KL} E.\ Knill and R.\ LaFlamme, ``On the power of one bit of quantum information,'' {\em Physical Review Letters} 81(25), pp.\ 5672--5675
, 1998. 
\bibitem[LLZ$^+$'23]{LLZ+} S.\ Lee, J.\ Lee, H.\ Zhai, Y.\ Tong, A.\ M.\ Dalzell, A.\ Kumar, P.\ Helms, J.\ Gray, Z.-H.\ Cui, W.\ Liu, M.\ Kastoryano, R.\ Babbush, J.\ Preskill, D.\ R.\ Reichman, E.\ T.\ Campbell, E.\ F.\ Valeev, L.\ Lin, and G.\ K.-L.\ Chan, ``Evaluating the evidence for exponential quantum advantage in ground-state quantum chemistry,'' {\em Nature Communications} 14(1), 2023.
\bibitem[LW'22]{LW} X.\ Li and C.\ Wang, ``Simulating Markovian open quantum systems using higher-order series expansion,''
{\em International Colloquium on Automata, Languages, and Programming}, 2023.
\bibitem[Lin'22]{Lin} L.\ Lin, \href{https://math.berkeley.edu/~linlin/qasc/qasc_notes.pdf}{\em Lecture Notes on Quantum Algorithms for Scientific Computation}, 2022. 
\bibitem[LT'20]{LT} L.\ Lin and Y.\ Tong, ``Near-optimal ground state preparation,'' {\em Quantum} 4, 2020.
\bibitem[OBK$^+$'16]{OBK+} P.\ J.\ J.\ O’Malley, R.\ Babbush, I.\ D.\ Kivlichan, J.\ Romero, J.\ R.\ McClean, R.\ Barends, J.\ Kelly,
P.\ Roushan, A.\ Tranter, N.\ Ding, B.\ Campbell, Y.\ Chen,
Z.\ Chen, B.\ Chiaro, A.\ Dunsworth, A.\ G.\ Fowler, E.\ Jeffrey, E.\ Lucero, A.\ Megrant, J.\ Y.\ Mutus, M.\ Neeley,
C.\ Neill, C.\ Quintana, D.\ Sank, A.\ Vainsencher, J.\ Wenner, T.\ C.\ White, P.\ V.\ Coveney, P.\ J.\ Love, H.\ Neven,
A.\ Aspuru-Guzik, and J.\ M.\ Martinis, ``Scalable quantum
simulation of molecular energies,'' {\em Physical Review X} 6, 2016. 
\bibitem[Ros'24]{Ros} G.\ Rosenthal, ``Efficient Quantum State Synthesis with One Query,'' {\em ACM--SIAM Symposium on Discrete Algorithms}, 2024. 
\bibitem[RFA'24]{RFA} C.\ Rouzé, D.\ S.\ Fran\c{c}a, and A.\ M.\ Alhambra, ``Efficient thermalization and universal quantum computing with quantum Gibbs samplers,'' \arxiv{2403.12691}, 2024. 
\bibitem[RM'23]{SRM} S.\ Roy Moulik, personal communication, 2023. 
\bibitem[She'06]{Shepherd} D.\ J.\ Shepherd, ``Computation with Unitaries and One Pure Qubit,'' \arxiv{quant-ph/0608132}, 2006. 
\bibitem[SJ'08]{SJ} P.\ W.\ Shor and S.\ P.\ Jordan, ``Estimating Jones polynomials is a complete problem for one clean qubit,'' {\em Quantum Information and Computation} 8(8), pp.\ 681--714, 2008. 
\bibitem[SM'23]{MS} O.\ Shtanko and R.\ Movassagh, ``Preparing thermal states on noiseless and noisy programmable quantum processors,'' \arxiv{2112.14688}, 2023. 
\bibitem[TOV$^+$'11]{TOV+} K.\ Temme, T.\ J.\ Osborne, K.\ G.\ Vollbrecht, D.\ Poulin, and F.\ Verstraete, ``Quantum Metropolis sampling,'' {\em Nature} 471(7336), pp.\ 87--90, 2011.
\bibitem[YA'12]{YA} M.-H.\ Yung and A.\ Aspuru-Guzik, ``A quantum--quantum Metropolis algorithm,'' {\em Proceedings of the National Academy of Sciences} 109(3), pp.\ 754--759, 2012.
}
\end{thebibliography}
\end{document}